\documentclass[oneside,english]{amsart}
\usepackage[T1]{fontenc}
\usepackage[utf8]{inputenc}
\usepackage{amsthm}
\usepackage{amssymb}

\makeatletter
\numberwithin{equation}{section}
\numberwithin{figure}{section}
\theoremstyle{plain}
\newtheorem{thm}{\protect\theoremname}
  \theoremstyle{definition}
  \newtheorem{defn}[thm]{\protect\definitionname}

\makeatother

\usepackage{babel}
  \providecommand{\definitionname}{Definition}
\providecommand{\theoremname}{Theorem}

\begin{document}

\title{Solvability of HornSAT and CNFSAT}

\author{KOBAYASHI, Koji}

\maketitle

\section{Overview}

This article describes the solvability of HornSAT and CNFSAT. Unsatisfiable
HornCNF have partially ordered set that is made by causation of each
clauses. In this partially ordered set, Truth value assignment that
is false in each clauses become simply connected space. Therefore,
if we reduce CNFSAT to HornSAT, we must make such partially ordered
set in HornSAT. But CNFSAT have correlations of each clauses, the
partially ordered set is not in polynomial size. Therefore, we cannot
reduce CNFSAT to HornSAT in polynomial size.

\section{causation of HornSAT}

We show the simply connectivity of $HornSAT$. Because $HornSAT$
is made by causation of each clauses, we can use unit resolution at
$HornSAT$. Therefore, $HornSAT$ have partially ordered set that
become simply connected space.
\begin{defn}
\label{def: Normal series} In formula $F\in CNF$, we will use the
term ``Metric space of $F$'' as the metric space of truth value
assignments that are false in $F$. In partially ordered set that
element is subset of clauses in CNF, we think partially ordered set
$\left\{ f_{i},<\right\} $ that different of successor is one clause; 

$f_{p}<f_{q}\rightarrow f_{p}\wedge c=f_{q},\, f_{p},f_{q}\subset F,c\in F$

I will use the term ``Normal series of $F$'' as $\left\{ f_{i},<\right\} $
if all $f_{i}$ become simply connected space.

In addition, we show normal series as the partially ordered set of
clauses of difference elements. That is, if normal series is $f_{p}<f_{q}<f_{r}$
and $f_{q}\setminus f_{p}=c_{pq}$,$f_{r}\setminus f_{q}=c_{qr}$,
we show normal series as $c_{pq}<c_{qr}$.\end{defn}
\begin{thm}
\label{thm: Simply connectivity of HornSAT} $F\in\overline{HornSAT}$
have normal series. Size of normal series at most $F$ size. \end{thm}
\begin{proof}
We can use unit resolution at $\overline{HornSAT}$. Therefore, we
can resolute each variables from each clauses by using unit resolution.
Each element $f_{i}$ of this partially ordered set become true in
truth value assignment that all value are true, and become false in
other truth value assignments. That is, all $f_{i}$ become simply
connected space. And size of this normal series is the size of unit
resolution that necessary to resolute empty clause. Therefore, this
theorem is proved.
\end{proof}

\section{HornSAT and CNFSAT}

We show the $HornSAT$ that include $CNFSAT$. Using this $HornSAT$,
We can deal with $CNFSAT$ in $HornSAT$ structure.
\begin{thm}
\label{thm: HornSAT that include CNFSAT} If $CNFSAT$ is in P, polynomial
size $HornSAT$ that include $CNFSAT$ is exists.\end{thm}
\begin{proof}
Think a series of computation that reduce $CNFSAT$ to $HornSAT$
and compute result. We can compute this computation with DTM. Therefore,
$HornSAT$ is exist that emulate this DTM computation. Because This
DTM computation include $CNFSAT$ problem as the start configuration,
$HornSAT$ that emulate DTM computation include $CNFSAT$ problem.
And if $CNFSAT$ is in P, DTM computation and $HornSAT$ that emulate
DTM computation are also in P. Therefore, this theorem is proved.
\end{proof}

\section{Correlations of CNFSAT}

We show that reduction $\overline{CNFSAT}$ to $\overline{HornSAT}$.
Mentioned above \ref{thm: Simply connectivity of HornSAT}, $\overline{HornSAT}$
have normal series. Therefore $\overline{HornSAT}$ that reduce $\overline{CNFSAT}$
have normal series. But $\overline{CNFSAT}$ have correlations of
clauses, we cannot reduce $\overline{CNFSAT}$ to $\overline{HornSAT}$
with polynomial size.
\begin{thm}
\label{thm: CNFSAT normal series} In $F\in\overline{CNFSAT}$, there
are some pair of clauses $\left(c_{a0},c_{a1}\right)$, $\left(c_{b0},c_{b1}\right)$,
$\left(c_{c0},c_{c1}\right)$,$\cdots$ that are disconnected each
other. $f_{a},f_{b},f_{c},\cdots$ is separated space of $\left(c_{a0},c_{a1}\right)$,
$\left(c_{b0},c_{b1}\right)$, $\left(c_{c0},c_{c1}\right)$,$\cdots$
. And $c_{a},c_{b},c_{c},\cdots$, $c_{ab},c_{ac},c_{bc},\cdots$,
$c_{abc},\cdots$ is clauses that become false if and only if in some
$f_{a},f_{b},f_{c},\cdots$ power set $\left(f_{a}\right)$,$\left(f_{b}\right)$,$\left(f_{c}\right)$,$\cdots$,
$\left(f_{a}\vee f_{b}\right)$, $\left(f_{b}\vee f_{c}\right)$,
$\left(f_{a}\vee f_{c}\right)$,$\cdots$, $\left(f_{a}\vee f_{b}\vee f_{c}\right)$,$\cdots$
space.

In this condition, normal series of $F$ include partially ordered
set $c_{a},c_{b},c_{c},\cdots$, $c_{ab},c_{ac},c_{bc},\cdots$, $c_{abc},\cdots$.
That is;

$\vdots$

$\cdots\, c_{abc}<c_{ab}<c_{a}$

$\cdots\, c_{abc}<c_{ac}<c_{a}$

$\cdots\, c_{abc}<c_{ab}<c_{b}$

$\cdots\, c_{abc}<c_{bc}<c_{b}$

$\cdots\, c_{abc}<c_{bc}<c_{c}$

$\cdots\, c_{abc}<c_{ac}<c_{c}$

$\vdots$\end{thm}
\begin{proof}
I prove it using reduction to absurdity. We assume that normal series
of $F$ is exist that not include partially ordered set $c_{a},c_{b},c_{c},\cdots$,$c_{ab},c_{ac},c_{bc},\cdots$,$c_{abc},\cdots$
in theorem condition.

We think the case that does not include $c_{ab}$ in normal series.
From the condition of theorem, the space exist that only $c_{ab}$
become false. Therefore, if normal series do not include $c_{ab}$,
$c_{ac}<c_{a}$ cannot keep normal series condition that $c_{ac},c_{a}$
difference is one clause and contradicts a condition. This condition
are not only $c_{ab}$ but also any other clauses.

Therefore, this theorem was shown than reduction to absurdity.\end{proof}
\begin{thm}
\label{thm: Product of Clauses} We can make some clauses that become
false if and only if in some $f_{a},f_{b},\cdots$ power set $\left(f_{a}\right)$,$\left(f_{b}\right)$,$\cdots$,$\left(f_{a}\vee f_{b}\right)$,$\cdots$
space in polynomial size of $f_{a},f_{b},\cdots$ .\end{thm}
\begin{proof}
$c_{a},c_{b},\cdots$ is the clauses that become false if and only
if in some $\left(f_{a}\right)$,$\left(f_{b}\right)$,$\cdots$ space,
and $c_{A},c_{B},\cdots$ is the clause that become true in all $\left(f_{a}\right)$,$\left(f_{b}\right)$,$\cdots$
space. We can make proper CNF with parameter $x_{0},x_{1},x_{2},\cdots$;

$\left(c_{a}\vee x_{0}\right)\wedge\left(c_{A}\vee x_{0}\right)\wedge\left(\overline{x_{0}}\vee c_{b}\vee x_{1}\right)\wedge\left(\overline{x_{0}}\vee c_{B}\vee x_{1}\right)\wedge\cdots$

$=\left(c_{a}\vee c_{b}\vee\cdots\right)\wedge\left(c_{A}\vee c_{b}\vee\cdots\right)\wedge\left(c_{a}\vee c_{B}\vee\cdots\right)\wedge\left(c_{A}\vee c_{B}\vee\cdots\right)\wedge\cdots$

The CNF become false if and only if in $\left(f_{a}\right)$,$\left(f_{b}\right)$,$\cdots$,$\left(f_{a}\vee f_{b}\right)$,$\cdots$
space. And this CNF is at most twice size of $f_{a},f_{b},\cdots$.
Therefore, this theorem is proved.\end{proof}
\begin{thm}
\label{thm: Size of CNFSAT normal series}Normal series of some $\overline{CNFSAT}$
problem become over polynomial size.\end{thm}
\begin{proof}
Mentioned above \ref{thm: CNFSAT normal series}, some $\overline{CNFSAT}$
problem have normal series that include power set of pair of disconnection
clauses. And mentioned above \ref{thm: Product of Clauses}, we can
make these problem at most constant size of pair of disconnection
clauses. Therefore, this theorem is proved.\end{proof}
\begin{thm}
\label{thm: Complexity of CNFSAT}$\overline{CNFSAT}\nleq_{p}\overline{HornSAT}$\end{thm}
\begin{proof}
Mentioned above \ref{thm: Size of CNFSAT normal series}, some $\overline{CNFSAT}$
problem have over polynomial size normal series. But mentioned above\ref{thm: HornSAT that include CNFSAT},
if $CNFSAT$ is in P, we can make $HornSAT$ that include $CNFSAT$
problem in polynomial size. And mentioned above \ref{thm: Simply connectivity of HornSAT},
size of normal series of $\overline{HornSAT}$ is at most formula
size and contradicts a condition that $CNFSAT$ is in P.

Therefore, we cannot $\overline{CNFSAT}$ to $\overline{HornSAT}$
in polynomial size.\end{proof}

\end{document}